%% file: interp.tex
\documentclass[letterpaper,11pt]{article}
\pdfoutput=1

\usepackage{amsmath,amsthm,amssymb}

\usepackage[%
  pdftitle={Diversification improves interpolation},
  pdfauthor={Mark Giesbrecht and Daniel S. Roche},
  colorlinks,linkcolor=darkblue,citecolor=darkgreen,urlcolor=darkred%
]{hyperref}

\usepackage[svgnames]{xcolor}
\definecolor{darkgreen}{rgb}{0,.35,0}
\definecolor{darkblue}{rgb}{0,0,.5}
\definecolor{darkred}{rgb}{.6,0,0}

\DeclareMathAlphabet{\mathbold}{OML}{cmm}{b}{it}

\RequirePackage[round]{natbib}

\usepackage{comment}
\usepackage[center,bf]{caption}

\usepackage[ruled,vlined,linesnumbered]{algorithm2e}
\DontPrintSemicolon

\numberwithin{equation}{section}
\theoremstyle{plain}
\newtheorem{theorem}{Theorem}[section]
\newtheorem{lemma}[theorem]{Lemma}
\newtheorem{definition}[theorem]{Definition}

\newcommand{\tinyspace}{\mspace{1mu}}

\newcommand{\gf}[1]{\ensuremath{\mathbb{F}_{#1}}}
\newcommand{\M}{\ensuremath{\mathsf{M}}}

\newcommand{\E}{\ensuremath{\mathsf{E}}}
\newcommand{\F}{\ensuremath{\mathsf{F}}}
\newcommand{\NN}{\ensuremath{\mathbb{N}}}
\newcommand{\ZZ}{\ensuremath{\mathbb{Z}}}

\newcommand{\RR}{\ensuremath{\mathbb{R}}}
\newcommand{\CC}{\ensuremath{\mathbb{C}}}
\newcommand{\softO}{{O\mskip1mu\tilde{\,}\mskip1mu}}

\newcommand{\norm}[1]{\left\lVert\tinyspace#1\tinyspace\right\rVert_2}
\newcommand{\ndivs}{\nmid}
\newcommand{\abs}[1]{\left\lvert\tinyspace #1 \tinyspace\right\rvert}
\newcommand{\ee}{\mathbf{e}}
\newcommand{\ii}{\mathbf{i}}

\newcommand{\expnot}[2]{#1\,e$-#2$}

\DeclareMathOperator{\rem}{rem\mskip2mu}
\DeclareMathOperator{\llog}{loglog}

\begin{document}

\title{Diversification improves interpolation}
\author{Mark Giesbrecht \and Daniel S. Roche}

\maketitle

\begin{abstract}
\input{abstract}
\end{abstract}

\input{intro}
\input{generic}
\input{ff}
\input{approx}
\input{impl}

\input{concl}
\input{ack}

\bibliographystyle{plainnat}
\bibliography{interp}

\end{document}

%% file: abstract.tex
We consider the problem of interpolating an unknown multivariate
polynomial with coefficients taken from a finite field or as
numerical approximations of complex numbers. Building on the
recent work of Garg and Schost, we improve on the best-known
algorithm for interpolation over large finite fields by presenting
a Las Vegas randomized algorithm that uses fewer black box
evaluations. Using related techniques, we also address numerical
interpolation of sparse polynomials with complex coefficients, and
provide the first provably stable algorithm (in the sense of
relative error) for this problem, at the cost of modestly more
evaluations.  A key new technique is a randomization which makes
all coefficients of the unknown polynomial distinguishable,
producing what we call a \emph{diverse} polynomial.  Another
departure from most previous approaches is that our algorithms do
not rely on root finding as a subroutine. We show how these
improvements affect the practical performance with trial
implementations.


%% file: intro.tex
\section{Introduction}

Polynomial interpolation is a long-studied and important problem in
computer algebra and symbolic computation. Given a way to evaluate an
unknown polynomial at any chosen point, and an upper bound on the
degree, the interpolation problem is to determine a representation for
the polynomial. In \emph{sparse} interpolation, we are also given an
upper bound on the number of nonzero terms in the unknown polynomial,
and the output is generally returned in the sparse (also
\emph{lacunary} or \emph{supersparse}) representation, wherein only
the nonzero terms are explicitly stored.

Applications of sparse interpolation include the
manipulation and factorization of multivariate polynomials and system
solving (see, e.g.,
\cite{CanKalLak89,KalTra90,DiaKal95,DiaKal98,JavMon07,JavMon09}.  With
the advent of hybrid symbolic-numeric algorithms for (systems of)
multivariate polynomials with approximate coefficients, we find
applications of approximate sparse interpolation, in particular for
solving non-linear systems of equations (see, e.g.,
\cite{SomVerWam01,SomVerWam04,Ste04}) and factoring approximate
multivariate polynomials (see, e.g., \cite{KalMayYanZhi08}).

Sparse interpolation is also a non-trivial generalization of the
important problem of \emph{polynomial identity testing}: given a
black box (especially an algebraic circuit) computing an unknown 
polynomial, determine whether the polynomial is zero. A relevant result
in our setting of sparse polynomials is \cite{BHLV09}; for a more
in-depth discussion, we recommend the recent survey by
\cite{Sax09}.

Here we examine the sparse interpolation problem in two settings which
have received recent attention: when the coefficients are elements of
finite fields (particularly large finite fields, over which we have no
choice) and when they are
approximations to complex numbers. We give improvements over the state
of the art in both of these cases, and demonstrate our new algorithms
in practice with a full implementation in C++.

\subsection{Problem definition}

Let $\F$ be a field.  A multivariate polynomial
$f\in\F[x_1,\ldots,x_n]$ is said to be $t$-sparse for some $t\in\NN$
if $f$ has at most $t$ nonzero terms in the standard power basis; that
is, $f$ can be written
\[
f = \sum_{i=1}^t c_i x_1^{e_{i1}} x_2^{e_{i2}} \cdots
x_n^{e_{in}} 
\] 
for coefficients $c_i\in\F$ and exponent tuples
$(e_{i1},\ldots,e_{in})\in\NN^n$ for $1\leq i\leq t$.
If each $e_i < d$, then the size of this representation is
$O(t)$ field elements plus $O(t n \log d)$ bits.
We seek algorithms which are
polynomial-time in the size of this representation.

Let $f\in\F[x_1,\ldots,x_n]$ have degree less than $d$. A \emph{black box} for
$f$ is a function which takes as input a vector
$(a_1,\ldots,a_n)\in\F^n$ and produces $f(a_1,\ldots,a_n)\in\F$.  The
cost of the black box is the number of operations in $\F$ required to
evaluate it at a given input.

\cite{CDGK91} showed that, if only evaluations over the ground field
$\F$ are allowed, then for some instances at least 
$\Omega(n^{\log t})$ black box probes are required. Hence if we seek
polynomial-time algorithms, we must extend the capabilities of the black
box.
To this end, \cite{DiaKal98} introduced the idea of an \emph{extended
  domain black box} which is capable of evaluating $f(b_1,\ldots,b_n)\in\E$ for
any $(b_1,\ldots,b_n)\in\E^n$ where $\E$ is any extension field of $F$.
That is, we can change every operation in the black box to work over an
extension field, usually paying an extra cost per evaluation
proportional to the size of the extension.

Motivated by the case of black boxes that are
division-free algebraic circuits, we will use the following 
model which we believe to be fair and cover all previous relevant
results. Here and for the remainder, $\M(m)$ is the number of field
operations required to multiply two univariate
polynomials with degrees less than
$m$, and $\softO(m)$ represents any function bounded by 
$m(\log m)^{O(1)}$. Using \cite{CK91}, 
$\M(m)\in O(m\log m\llog m)$, which is $\softO(m)$.

\begin{definition}\label{def:bb}
  Let $f\in\F[x_1,\ldots,x_n]$ and $\ell>0$.
  A \emph{remainder black box} for $f$ with size $\ell$ is a procedure
  which, given any monic square-free polynomial $g\in\F[y]$ with $\deg g=m$,
  and any $h_1,\ldots,h_n\in\F[y]$ with each $\deg h_i < m$,
  produces $f(h_1,\ldots,h_n) \rem g$ using at most
  $\ell \cdot \M(m)$ operations in $\F$.
\end{definition}

This definition is general enough to cover the algorithms we know of
over finite fields, and we submit that the cost model is fair to the
standard black box, extended domain black box, and algebraic circuit
settings. The model makes sense over complex numbers,
as we will see.


\subsection{Interpolation over finite fields}

We first summarize previously known univariate interpolation algorithms 
when $\F$ is a finite field with $q$ elements and identify our new
contributions here. For now, let $f\in\gf{q}[x]$ have degree less than
$d$ and
sparsity $t$. We will assume we have a remainder black box for $f$ with
size $\ell$. Since field elements can be represented with $O(\log q)$
bits, a polynomial-time algorithm will have cost polynomial in $\ell$,
$t$, $\log d$, and $\log q$.

For the dense output representation, one can use the classical method of
Newton/Waring/Lagrange to interpolate in $\softO(\ell d)$ time
\citep[\S 10.2]{vzGG03}.

The algorithm of \citet{BT88} for sparse polynomial interpolation,
and in particular the version developed by \citet{KY89}, can be
adapted to arbitrary finite fields. Unfortunately, these algorithms
require $t$ discrete logarithm computations in $\gf{q}^*$, whose cost is
small if the field size $q$ is chosen carefully (as in \citet{Kal10a}),
but not in general. For arbitrary (and potentially large) $q$, we can
take advantage of the fact that each discrete logarithm that needs to be
computed falls in the range $[0,1,\ldots,d-1]$.
The ``kangaroo method'' of \citet{Pol78,Pol00} can, with high
probability, compute such a discrete log with $O(\sqrt{d})$ field
operations.
Using this algorithm makes brings the total worst-case cost of Ben-Or
and Tiwari's algorithm to
$O(t\ell + t^2 + t\sqrt{d})$.

The current study builds most directly on the work of \cite{GS09}, who
gave the first polynomial-time algorithm for sparse interpolation over
an arbitrary finite field.
Their algorithm works roughly as follows. For
very small primes $p$, use the black box to compute $f$ modulo
$x^p-1$. A prime $p$ is a ``good prime'' if and only if all the terms of
$f$
are still distinct modulo $x^p-1$. 
If we do this for all $p$ in the range of roughly
$O(t^2\log d)$, then there will be sufficient good primes to recover the
unique symmetric polynomial over $\ZZ[y]$ whose roots are the exponents of
nonzero terms in $f$. We then factor this polynomial to find those
exponents, and correlate with any good prime image to
determine the coefficients. The total cost is
$\softO(\ell t^4 \log^2 d)$ field operations. Using randomization, it is
easy to reduce this to $\softO(\ell t^3 \log^2 d)$.

Observe that the coefficients of the symmetric integer polynomial in
Garg \& Schost's algorithm are bounded by $O(d^t)$, which is much larger than the
$O(d)$ size of the exponents ultimately recovered.
Our primary contribution over finite fields of size
at least $\Omega(t^2 d)$
is a new algorithm which
avoids evaluating the symmetric polynomial and performing root finding
over $\ZZ[y]$. As a result, we reduce the total number of required
evaluations and develop a randomized algorithm with cost
$\softO(\ell t^2 \log^2 d)$, which is roughly quadratic in the input
and output sizes. Since this can be deterministically verified in 
the same time, our algorithm (as well as the randomized version of Garg
\& Schost) is of the Las Vegas type.

\begin{table}[tp]
\noindent\makebox[\textwidth]{
\begin{tabular*}{1.2\textwidth}{@{\extracolsep{\fill}}c|c|c|c|c}
  & Probes & Probe degree & Computation cost & Total cost \\
  \hline
  Dense & $d$ & $1$ & $\softO(d)$ & $\softO(\ell d)$ \\
  Ben-Or \& Tiwari & $O(t)$ & $1$ & $O(t^2 + t\sqrt{d})$ &
    $\softO(\ell t + t^2 + t\sqrt{d})$ \\
  Garg \& Schost & $\softO(t^2\log d)$ & $\softO(t^2\log d)$ &
    $\softO(t^4 \log^2 d)$ & $\softO(\ell t^4 \log^2 d)$ \\
  Randomized G \& S & $\softO(t \log d)$ & $\softO(t^2 \log d)$ &
    $\softO(t^3\log^2 d)$ & $\softO(\ell t^3\log^2 d)$ \\
  Ours & $O(\log d)$ & $\softO(t^2\log d)$ & $\softO(t^2\log^2 d)$ &
    $\softO(\ell t^2\log^2 d)$ \\
\end{tabular*}}
\caption{Sparse univariate interpolation over large finite
fields,\\
with black box size $\ell$, degree $d$, and $t$ nonzero terms%
\label{table:prevff}}
\end{table}

The relevant previous results mentioned above
are summarized in
Table~\ref{table:prevff}, where we assume in all cases that the field size $q$
is ``large enough''. In the table, the ``probe degree'' refers to
the degree of $g$ in each evaluation of the remainder black box as
defined above.

\subsection{Multivariate interpolation}

Any of the univariate algorithms above can be used to generate a
multivariate polynomial interpolation algorithm in at least two
different ways. For what follows, write $\rho(d,t)$ for the number of
remainder black box evaluations required by some univariate
interpolation algorithm, $\Delta(d,t)$ for the degree of the remainder
in each evaluation, and $\psi(d,t)$ for the number of other field
operations required besides black box calls. Observe that these
correspond to the first three columns in Table~\ref{table:prevff}.

The first way to adapt a univariate interpolation algorithm to a
multivariate one is Kronecker
substitution: given a remainder black box for an unknown 
$f\in\F[x_1,\ldots,x_n]$, with each partial degree less than $d$,
we can easily construct a remainder black box for the univariate
polynomial
$\hat f = f(x, x^d, x^{d^2},\ldots,x^{d^{n-1}})\in\F[x]$, whose terms
correspond one-to-one with terms of $f$.  
This is the approach taken for instance in 
\citet[][\S 2]{Kal10a} for the interpolation of multivariate polynomials
with rational coefficients.
The cost is simply the cost of
the chosen underlying univariate algorithm, with the degree increased to
$d^n$.

The other method for constructing a multivariate interpolation algorithm
is due to \citet{Zip90}. The technique is inherently
probabilistic and works variable-by-variable, at each step solving a 
number of $\hat t \times \hat t$ transposed Vandermonde systems, for
some $\hat t \leq t$. Specifically, each system is of the form
$A x = b$, where $A$ is a $\hat t \times \hat t$ matrix of scalars from
the coefficient field $\gf{q}$. The vector $v$ consists of the output of
$\hat t$ remainder black box evaluations, and so its elements are in
$\gf{q}[y]$, and the system must be solved modulo some $g\in\gf{q}[y]$,
as specified by the underlying univariate algorithm. Observe however
that since $A$ does not contain polynomials, computing $x=A^{-1}b$
requires no modular polynomial arithmetic. In fact, using the same
techniques as \citet[][\S 5]{KY89}, employing fast dense bivariate
polynomial arithmetic, each system can be solved using
$$O\Big(\M\big(t\cdot \Delta(d,t)\big)\cdot\log\big(t\cdot
\Delta(d,t)\big)\Big)$$ field operations.

Each transposed Vandermonde system gives the remainder black box
evaluation of each of $\hat t$ univariate polynomials that we are
interpolating in that step. The number of such systems that must be
solved is therefore $\rho(d,t)$, as determined by the underlying
univariate algorithm.
Finally, each of the $\hat t$ univariate
interpolations proceeds with the given evaluations. 
The total cost, over all iterations, is
$$\softO\big(\ell nt\cdot\Delta(d,t)\cdot\rho(d,t)\big)$$
field operations for the remainder black box evaluations, plus
$$\softO\big(nt\psi(d,t) + \ell nt\cdot\Delta(d,t)\big)$$
field operations for additional computation.
\citet{Zip90} used the dense algorithm
for univariate interpolation; using Ben-Or and Tiwari's algorithm
instead was studied by \citet{KalLee03}.

\begin{table}[tp]
\noindent\makebox[\textwidth]{
\begin{tabular}{c|c|c}
  & Kronecker & Zippel \\
  \hline
  Dense & $\softO(\ell d^n)$ & $\softO(\ell ntd)$ \\
  Ben-Or \& Tiwari & $\softO(\ell t + t^2 + t d^{n/2})$ &
    $\softO(nt^3 + nt^2\sqrt{d} + \ell nt^2)$ \\
  Garg \& Schost & $\softO(\ell n^2 t^4 \log^2 d)$ &
    $\softO(\ell nt^5 \log^2 d)$ \\
  Randomized G \& S & $\softO(\ell n^2 t^3 \log^2 d)$ &
    $\softO(\ell nt^4 \log^2 d)$ \\
  Ours & $\softO(\ell n^2 t^2 \log^2 d)$ &
    $\softO(\ell nt^3 \log^2 d)$ 
\end{tabular}}
\caption{Sparse multivariate interpolation over large finite
fields,\\
with black box size $\ell$, $n$ variables, degree $d$, and $t$ nonzero terms%
\label{table:mvar}}
\end{table}

Table \ref{table:mvar} summarizes the cost of the univariate algorithms
mentioned above applied to sparse multivariate interpolation over a
sufficiently large finite field, using Kronecker's and Zippel's methods.

For completeness, we mention a few more results on closely related
problems that do not have a direct bearing on the current study.
\citet{GKS90} give a parallel algorithm with small depth but which
is not competitive in our model due to the large number of processors
required.
A practical parallel version of Ben-Or and Tiwari's algorithm has been
developed by
\cite{JavMon10}.
\cite{KLW90} and \cite{AKP06} present modular algorithms for
interpolating polynomials with rational and integer coefficients.
However, their methods do not seem to 
apply to finite fields.

\subsection{Approximate Polynomial Interpolation}

In Section \ref{sec:approx} we consider the case of approximate sparse
interpolation.  Our goal is to provide both a numerically more robust
practical algorithm, but also the first algorithm which is provably
numerically stable, with no heuristics or conjectures.  We define an
``$\epsilon$-approximate black box'' as one which evaluates an unknown
$t$-sparse target polynomial $f\in\CC[x]$, of degree $d$, with
relative error at most $\epsilon>0$.  Our goal is to build a $t$-sparse
polynomial $g$ such that $\norm{f-g}\leq\epsilon\norm{f}$.  A bound on
the degree and sparsity of the target polynomial, as well as
$\epsilon$, must also be provided.  In Section \ref{sec:approx} we
formally define the above problem, and demonstrate
that the problem of sparse interpolation is well-posed.  We then adapt
our variant of the \cite{GS09} algorithm for the approximate
case, prove it is numerically accurate in terms of the relative error of
the output, and analyze its cost.  We also present a full
implementation in Section \ref{sec:impl} and validating experiments.

Recently, a number of numerically-focussed sparse interpolation
algorithms have been presented.  The algorithm of \cite{GLL09} is a
numerical adaptation of \cite{BT88}, which samples $f$ at $O(t)$
randomly chosen roots of unity $\omega\in\CC$ on the unit circle. In
particular, $\omega$ is chosen to have (high) order at least the
degree, and a randomization scheme is used to avoid clustering of
nodes which will cause dramatic ill-conditioning.  A relatively weak
theoretical bound is proven there on the randomized conditioning
scheme, though experimental and heuristic evidence suggests it is much
better in practice.  \cite{CuyLee08} adapt Rutishauser's $qd$
algorithm to alleviate the need for bounds on the partial degrees and
the sparsity, but still evaluate at high-order roots of unity.
Approximate sparse rational function interpolation is considered by
\cite{KalYan07} and \cite{KalYanZhi07}, using the Structured Total
Least Norm (STLN) method and, in the latter, randomization to improve
conditioning.  Approximate sparse interpolation is also considered for
integer polynomials by \cite{Mansour95}, where a polynomial-time
algorithm is presented in quite a different model from ours.  In
particular the evaluation error is absolute (not relative) and the
complexity is sensitive to the bit length of the integer coefficients.

Note that all these works evaluate the polynomial only on the unit
circle.  This is necessary because we allow and expect $f$ to have
very large degree, which would cause a catastrophic loss of precision
at data points of non-unit magnitude.  Similarly, we assume that the
complex argument of evaluation points is exactly specified, which is again
necessary because any error in the argument would be exponentially
magnified by the degree.

The primary contribution of the work in this paper is to provide
an algorithm with both rigorously provable relative error and
good practical performance.  Our algorithm typically requires
$\softO(t^2\log^2d)$ evaluations at primitive roots of unity of order
$\softO(t^2\log d)$ (as opposed to order $d$ in previous approaches).
We guarantee that it finds a $t$-sparse polynomial $g$ such that
$\norm{g-f}\leq 2\epsilon\norm{f}$.  An experimental demonstration of
the numerical robustness is given in Section \ref{sec:impl}.


%% file: generic.tex
\section{Sparse interpolation for generic fields}
\label{sec:generic}

Here and for the remainder, we say a polynomial $f$
is \emph{$t$-sparse} if it can be written as a sum of at most $t$
nonzero coefficients times a monomial.
We assume the unknown polynomial $f$ is
always univariate. This is without loss of generality, as we can use the
Kronecker substitution as discussed above. 
The exponential
increase in the univariate degree only corresponds to a factor of $n$
increase in $\log \deg f$, and since our algorithms will ultimately have
cost polynomial in $\log \deg f$, polynomial time is preserved. 

Assume a fixed, unknown, $t$-sparse univariate polynomial $f\in\F[x]$
with
degree at most $d$. We will use a remainder black box for $f$ to evaluate
$f \rem (x^p-1)$ for small primes $p$. We say $p$ is a ``good prime''
if the sparsity of $f \rem (x^p-1)$ is the same as that of $f$ itself
--- that is, none of the exponents are equivalent modulo $p$.

The following lemma shows the size of primes required to randomly choose
good primes with high probability.

\begin{lemma}\label{lem:goodp}
  Let $f\in\F[x]$ be a $t$-sparse polynomial
  with degree $d$,
  and let 
  $\lambda = \max\left(21,\left\lceil \frac{5}{3}t(t-1)\ln d
  \right\rceil\right)$.
  A prime chosen at random in the range
  $\lambda,\ldots,2\lambda$
  is a good prime for $f$ with probability at least $1/2$.
\end{lemma}
\begin{proof}
  Write $e_1,\ldots,e_t$ for the exponents of nonzero terms in $f$.
  If $p$ is a bad prime, then $p$ divides $(e_j-e_i)$ for some $i<j$.
  Each $e_j-e_i \leq d$, so there can be at most
  $\log_\lambda d = \ln d / \ln \lambda$ primes that divide each
  $e_j-e_i$. There are exactly $\binom{t}{2}$ such pairs of exponents, so
  the total number of bad primes is at most
  $(t(t-1)\ln d)/(2\ln\lambda)$.

  From \citet[Corollary 3 to Theorem 2]{RosSch62},
  the total number of primes in the range $\lambda,\ldots,2\lambda$ is
  at least $3\lambda/(5\ln \lambda)$ when $\lambda \geq 21$, which is at
  least $t(t-1)\ln d/\ln \lambda$, at least
  twice the number of bad primes.
\end{proof}

Now observe an easy case for the sparse interpolation problem.
If a polynomial $f\in\F[x]$, has all coefficients distinct; that is,
$f=\sum_{1\leq i\leq t} c_ix^{e_i}$ and $c_i=c_j \Rightarrow i=j$, then
we say $f$ is \emph{diverse}. 
To interpolate a diverse polynomial $f\in\F[x]$, 
we first follow the method of
\cite{GS09} by computing $f \rem (x^{p_i} - 1)$ for ``good primes''
$p_i$ such that the sparsity of $f\rem (x^{p_i}-1)$ is the same as that
of $f$. Since $f$ is diverse, $f \rem (x^{p_i}-1)$ is also diverse and
in fact each modular image has the same set of coefficients. Using this
fact, we avoid the need to construct and subsequently factor the
symmetric polynomial in the exponents. Instead, we correlate like terms
based on the (unique) coefficients in each modular image, then use
simple Chinese remaindering to construct each exponent $e_i$ from its
image modulo each $p_i$. This requires only $O(\log d)$ 
remainder black box evaluations at
good primes, gaining a factor of $t$ improvement over the randomized
version of \cite{GS09}
for diverse polynomials.

In the following sections, we will show how to choose an $\alpha\in\F$
so that $f(\alpha x)$ --- which we can easily construct a remainder black
box for --- is diverse. With such a procedure,
Algorithm~\ref{alg:generic}
gives a Monte Carlo algorithm for interpolation over a general field.

\begin{algorithm}[Htb] 
  \caption{Generic interpolation\label{alg:generic}}
  \KwIn{$\mu\in\RR_{>0}$, 
    $T,D,q\in\NN$, and a remainder black box for unknown $T$-sparse
    $f\in\F[x]$ with $\deg f < D$}
  \KwOut{$t\in\NN$, $e_1,\ldots,e_t\in\NN$, and
    $c_1,\ldots,c_t\in\F$ such that $f=\sum_{1\leq i\leq
    t}c_ix^{e_i}$}
  $t \gets 0$\;
  $\lambda \gets \max\left(21,\left\lceil
  \frac{5}{3}T(T-1)\ln D \right\rceil\right)$\; 
  \For{$\lceil \log_2 (3/\mu) \rceil$ primes 
    $p \in \{\lambda,\ldots,2\lambda\}$ \label{step:ffgets}}
  {
    Use black box to compute $f_p = f(x) \rem (x^p-1)$\;
    \If{$f_p$ has more than $t$ terms}{
      $t \gets$ sparsity of $f_p$\;
      $\varrho \gets p$\;
    }
  }
  $\alpha \gets$ element of $\F$ such that 
    $\Pr[f(\alpha x)\text{ is not diverse}] < \mu/3$ \label{step:diversify} \;
  $g_\varrho \gets f(\alpha x) \rem (x^\varrho - 1)$\;
  $c_1,\ldots,c_t \gets$ nonzero coefficients of $g_{\varrho}$ \;
  $e_1,\ldots,e_t \gets 0$ \;
  \For{$\lceil 2 \ln (3/\mu) + 4(\ln D)/(\ln \lambda) \rceil$
    primes $p\in\{\lambda,\ldots,2\lambda\}$ \label{step:ffinterp}}
  {
    Use black box to compute $g_p = f(\alpha x) \rem (x^p-1)$ \;
    \If{$g_p$ has exactly $t$ nonzero terms}{
      \lFor{$i=1,\ldots,t$}{Update $e_i$ with exponent of $c_i$ in $g_p$
      modulo $p$ via Chinese remaindering} \;
    }
  }
  \lFor{$i=1,\ldots,t$}{$c_i \gets c_i \alpha^{-e_i}$} \;
  \KwRet{$ f(x) = \sum_{1 \leq i \leq t} c_i x^{e_i}$} \;
\end{algorithm}

\noindent\begin{minipage}{\textwidth}
\begin{theorem}\label{thm:genalg}
  With inputs as specified, 
  Algorithm~\ref{alg:generic} correctly computes the unknown polynomial $f$
  with probability at least $1-\mu$.
  The total cost in field operations (except for
  step~\ref{step:diversify}) is
  $$O\left(
    \ell \cdot \left(
      \frac{\log D}{\log T + \llog D} + 
      \log \frac{1}{\mu}
    \right)
    \cdot \M\left(T^2 \log D\right)
  \right).$$
\end{theorem}
\end{minipage}

\begin{proof}
  The for loop on line~\ref{step:ffgets} searches for the true sparsity
  $t$ and a single good prime $\varrho$. Since each prime $p$ in the
  given range is good with probability at least $1/2$ by
  Lemma~\ref{lem:goodp}, the probability of failure at this stage is at
  most $\mu/3$.

  The for loop on line~\ref{step:ffinterp} searches for and
  uses sufficiently many good primes to recover the exponents of $f$.
  The product of all the good primes must be at least $D$, and since
  each prime is at least $\lambda$, at least $(\ln D)/(\ln \lambda)$ good
  primes are required.

  Let $n=\lceil 2 \ln (3/\mu) + 4(\ln D)/(\ln \lambda) \rceil$ be the
  number of primes sampled in this loop, and 
  $k=\lceil(\ln D)/(\ln \lambda)\rceil$
  the number of good primes required. We can derive that
  $(n/2-k)^2 \geq (\ln(3/\mu)+k)^2 > (n/2)\ln (3/\mu)$,
  and therefore
  $\exp(-2(\frac{n}{2}-k)^2/n) < \mu/3$.
  Using Hoeffding's Inequality \citep{Hoe63}, this means the probability
  of encountering fewer than $k$ good primes is less than $\mu/3$.

  Therefore the total probability of failure is at most $\mu$. For
  the cost analysis, the dominating cost will be the
  modular black box evaluations in the last for loop. 
  The number of evaluations in this loop is
  $O(\log(1/\mu) + (\log D)/(\log \lambda))$, 
  and each evaluation has cost $O(\ell \cdot \M(\lambda))$. Since the
  size of each prime is
  $\Theta((\log D)/(\log T + \llog D))$, the 
  complexity bound is correct as stated.
\end{proof}

In case the bound $T$ on the number of nonzero terms is very bad, we
could choose a smaller value of $\lambda$ based on the true sparsity $t$
before line~\ref{step:diversify}, improving the cost of the remainder of
the algorithm.

In addition, as our bound on possible number of ``bad
primes'' seems to be quite loose, a more efficient approach in practice
would be to replace the for loop on line~\ref{step:ffinterp} with one
that starts with a prime much smaller than $\lambda$
and incrementally searches for the next larger primes
until the product of all good primes is at least $D$. 
We could choose the lower bound to start searching from based on lower
bounds on the birthday problem. That is, assuming (falsely) that the
exponents are randomly distributed modulo $p$, start with the least $p$
that will have no exponents collide modulo $p$ with high probability.
This
would yield an algorithm more sensitive to the true bound on bad primes,
but unfortunately gives a worse formal cost analysis.


%% file: ff.tex
\section{Sparse interpolation over finite fields}
\label{sec:ff}

We now examine the case that the ground field $\F$ is the finite field
with $q$ elements, which we denote $\gf{q}$. First we show how to
effectively diversify the unknown polynomial $f$ in order to
complete Algorithm~\ref{alg:generic} for the case of large finite
fields. Then we show how to extend this to 
a Las Vegas algorithm with the same complexity.

\subsection{Diversification}
For an unknown $f\in\gf{q}[x]$ given by a remainder black box,
we must find an $\alpha$ so that $f(\alpha x)$ is diverse.
A surprisingly simple trick works: evaluating $f(\alpha x)$ for a
random nonzero $\alpha \in\gf{q}$. 

\begin{theorem}\label{thm:diverseff}
  For $q \geq T(T-1)T + 1$ and 
  any $T$-sparse polynomial $f\in\gf{q}[x]$ with $\deg f < D$,
  if $\alpha$ is chosen uniformly
  at random from $\gf{q}^*$, the probability that $f(\alpha x)$ is 
  diverse is at least $1/2$.
\end{theorem}
\begin{proof}
  Let $t\leq T$ be the exact number of nonzero terms in $f$, and
  write 
  $f = \sum_{1\leq i \leq t} c_i x^{e_i}$, with nonzero coefficients
  $c_i \in\gf{q}^*$ and $e_1 < e_2 < \cdots < e_t$.
  So the $i$th coefficient of $f(\alpha x)$ is $c_i\alpha^{e_i}$.

  If $f(\alpha x)$ is \emph{not} diverse, then we must have
  $c_i\alpha^{e_i} = c_j\alpha^{e_j}$ for some $i\neq j$. Therefore
  consider the polynomial $A\in\gf{q}[y]$ defined by
  $$A = \prod_{1\leq i < j \leq t} \left(c_i y^{e_i}- c_j
  y^{e_j}\right).$$
  We see that $f(\alpha x)$ is diverse if and only if $A(\alpha) \neq
  0$, hence the number of roots of $A$ over $\gf{q}$ is exactly the
  number of unlucky choices for $\alpha$.
  
  The polynomial $A$ is the product of exactly $\binom{t}{2}$ binomials,
  each of which has degree less than $D$. Therefore 
  $$\deg A < \frac{T(T-1)D}{2},$$
  and this also gives an upper bound on the number of roots of $A$.
  Hence $q-1 \geq 2 \deg A$, and at least half of the elements of
  $\gf{q}^*$ are not roots of $A$, yielding the stated result.
\end{proof}

Using this result, given a black box for $f$ and the exact sparsity $t$
of $f$, we can find an $\alpha\in\gf{q}$ such that $f(\alpha x)$ is
diverse by sampling 
 random values
$\alpha\in\gf{q}$, evaluating $f(\alpha x) \rem x^p-1$ for a single
good prime $p$, and checking whether the polynomial is diverse. 
With probability at least $1-\mu$, this will
succeed in finding a diversifying $\alpha$ after at most 
$\lceil \log_2 (1/\mu) \rceil$ iterations.
Therefore we can use this approach in Algorithm~\ref{alg:generic}
with no effect on the asymptotic complexity.

\subsection{Verification}
\label{sec:ffver}

So far,
Algorithm~\ref{alg:generic} over a finite field
is probabilistic of the Monte Carlo type; that
is, it may give the wrong answer with some controllably-small
probability. To provide a more robust Las Vegas probabilistic algorithm,
we require only a fast way to check that a candidate answer is in fact
correct. To do this, observe that given a modular black box for an
unknown $T$-sparse $f\in\gf{q}[x]$ and an explicit $T$-sparse polynomial
$g\in\gf{q}[x]$, we can construct a modular black box for the
$2T$-sparse polynomial $f-g$ of their difference. Verifying that $f=g$
thus reduces to the well-studied problem of deterministic polynomial
identity testing.

The following algorithm is due to
\cite{BHLV09} and provides this check in essentially
the same time as the interpolation algorithm; 
we restate it in Algorithm~\ref{alg:ver} for completeness and to use our
notation.

\newcommand{\ZERO}{\textbf{\upshape ZERO}}
\newcommand{\NZ}{\textbf{\upshape NONZERO}}
\begin{algorithm}
  \caption{Verification over finite fields\label{alg:ver}}
  \KwIn{$T,D,q\in\NN$ and remainder black box for unknown $T$-sparse
    $f\in\gf{q}[x]$ with $\deg f \leq D$}
  \KwOut{\ZERO\ iff $f$ is identically zero}
  \For{the least $(T-1) \log_2 D$ primes $p$}{
    Use black box to compute $f_p = f \rem (x^p - 1)$ \;
    \lIf{$f_p \neq 0$}{\KwRet{\NZ}}\;
  }
  \KwRet{\ZERO}\;
\end{algorithm}
\begin{theorem}\label{thm:ver}
  Algorithm~\ref{alg:ver} works correctly as stated and uses at most
  $$O\left(\ell T \log D \cdot 
    \M\left( T\log D\cdot\left(\log T + \llog D\right)\right)\right)$$
  field operations.
\end{theorem}
\begin{proof}
  For correctness, notice that the requirements for a ``good prime'' for
  identity testing are much weaker than for interpolation. Here, we only
  require that a single nonzero term not collide with any other nonzero
  term. That is, every bad prime $p$ will divide $e_j-e_1$ for some
  $2 \leq j \leq T$. There can be $\log_2 D$ distinct prime divisors of
  each $e_j-e_1$, and there are $T-1$ such differences. Therefore
  testing that the polynomial is zero modulo $x^p-1$ for the first
  $(T-1)\log_2 D$ primes is sufficient to guarantee at least one
  nonzero evaluation of a nonzero $T$-sparse polynomial.

  For the cost analysis, 
  the prime number theorem \citep[Theorem 8.8.4]{BS96},
  tells us that the first $(T-1)\log_2 D$ primes are each bounded by
  $O(T\cdot \log D\cdot (\log T + \llog D))$. The stated bound follows
  directly.
\end{proof}

This provides all that we need to prove the main result of this
section:

\begin{theorem}\label{thm:ffmain}
  Given $q\geq T(T-1)D + 1$,
  any $T,D\in\NN$, and a modular black box for unknown $T$-sparse
  $f\in\gf{q}[x]$ with $\deg f \leq D$, there is an algorithm
  that always produces the correct
  polynomial $f$ and with high probability uses only
  $\softO\left(\ell T^2\log^2 D \right)$
  field operations.
\end{theorem}
\begin{proof}
  Use Algorithms \ref{alg:generic}
  and \ref{alg:ver} with $\mu=1/2$, looping as necessary until
  the verification step succeeds. With high probability, only a constant
  number of iterations will be necessary, and so the cost is as stated.
\end{proof}

For the small field case, when $q \in O(T^2 D)$, 
the obvious approach would be to work in an extension $\E$ of size
$O(\log T + \log D)$ over $\gf{q}$. Unfortunately, this would presumably
increase the cost of each evaluation by a factor of $\log D$, potentially
dominating our factor of $T$ savings compared to the randomized version
of \cite{GS09} when the unknown polynomial has very few terms
and extremely high degree. 

In practice, it seems that a much smaller extension than this is
sufficient in
any case to make each $\gcd(e_j-e_i,q-1)$ small compared to
$q-1$, but we do not yet know how to prove any tighter bound in the worst
case.


%% file: approx.tex
\section{Approximate sparse interpolation algorithms}
\label{sec:approx}

In this section we consider the problem of interpolating an
approximate sparse polynomial $f\in\CC[x]$ from evaluations on the
unit circle.    We will generally assume that $f$ is $t$-sparse:
\begin{equation}
  \label{eq:f}
  f = \sum_{1\leq i\leq t} c_i x^{e_i}
  ~\mbox{for $c_i\in\CC$ and $e_1<\cdots <e_t=d$.}
\end{equation}
We require a notion of size for such polynomials, and
define the coefficient 2-norm of 
$f = \sum_{0\leq i\leq d} f_i\,x^i$ as
\[
\norm{f}=
\sqrt{\sum_{0\leq i\leq d} |f_i|^2}.
\]

The following identity relates the norm of evaluations on the unit
circle and the norm of the coefficients.
As in Section~\ref{sec:generic}, for $f\in\CC[x]$ is as in \eqref{eq:f}, we
say that a prime $p$ is a \emph{good prime} for $f$ if $p\ndivs
(e_i-e_j)$ for all $i\neq j$.
\begin{lemma}
  \label{lem:normtrans}
  Let $f\in\CC[x]$, $p$ a good prime for $f$, and $\omega\in\CC$ a $p$th
  primitive root of unity.  Then
  \[
  \norm{f}^2=\frac{1}{p}\sum_{0\leq i<p} \abs{f(\omega^i)}^2.
  \]
\end{lemma}
See \citet[Theorem 2.9]{GieRoc09}.

We can now formally define the approximate sparse univariate
interpolation problem.

\begin{definition}
Let $\epsilon>0$ and assume there exists an unknown $t$-sparse
$f\in\CC[x]$ of degree at most $D$.  An \emph{$\epsilon$-approximate
  black box} for $f$ takes an input $\xi\in\CC$ and produces a
$\gamma\in\CC$ such that $|\gamma-f(\xi)|\leq\epsilon|f(\xi)|$.
\end{definition}
That is,
the relative error of any evaluation is at most $\epsilon$.  As noted
in the introduction, we will specify our input points exactly, at
(relatively low order) roots of unity.
The \emph{approximate sparse univariate interpolation problem} is then as
follows:  given $D,T\in\NN$ and $\delta\geq\epsilon>0$, and an
$\epsilon$-approximate black box for an unknown $T$-sparse polynomial
$f\in\CC[x]$ of degree at most $D$, find a
$T$-sparse polynomial $g\in\CC[x]$ such that
$\norm{f-g}\leq\delta\norm{g}$.

The following theorem shows that $t$-sparse polynomials are
well-defined by good evaluations on the unit circle.
\begin{theorem} \label{thm:posed}
  Let $\epsilon>0$ and $f\in\CC[x]$ be a $t$-sparse polynomial.  Suppose
  there exists a $t$-sparse polynomial $g\in\CC[x]$ such that for a
  prime $p$ which is good for both $f$ and $f-g$, and $p$th primitive root
  of unity $\omega\in\CC$, we have
  \[
  |f(\omega^i)-g(\omega^i)|\leq\epsilon |f(\omega^i)|
  ~~~\mbox{for $0\leq i<p$}.
  \]
  Then $\norm{f-g}\leq\epsilon\norm{f}$.  Moreover, if $g_0\in\CC[x]$ is
  formed from $g$ by deleting all the terms not in the support of $f$,
  then $\norm{f-g_0}\leq2\epsilon\norm{f}$.
\end{theorem}
\begin{proof}
  Summing over powers of $\omega$ we have
  \[
  \sum_{0\leq i<p} |f(\omega^i)-g(\omega^i)|^2 \leq\epsilon^2 \sum_{0\leq
    i<p} |f(\omega^i)|^2.
  \]
  Thus, since $p$ is a good prime for both $f-g$ and $f$, 
  using Lemma~\ref{lem:normtrans}, $p\cdot
  \norm{f-g}^2\leq\epsilon^2\cdot p \cdot \norm{f}^2$ and
  $\norm{f-g}\leq\epsilon\norm{f}$.

  Since $g-g_0$ has no support in common with $f$,
  \[
  \norm{g-g_0}\leq \norm{f-g}\leq\epsilon\norm{f}.
  \]
  Thus
  \begin{align*}
  \norm{f-g_0} & =\norm{f-g+(g-g_0)} \\
     & \leq \norm{f-g}+\norm{g-g_0} \leq
  2\epsilon\norm{f}. \qed
  \end{align*}
\end{proof}
In other words, any $t$-sparse polynomial whose values are very close
to $f$ must have the same support except possibly for some terms with
very small coefficients.

\subsection{Computing the norm of an 
approximate sparse polynomial}

Let $0<\epsilon<1/2$ and $f\in\CC[x]$ a $t$-sparse polynomial for which we
are given an $\epsilon$-approximate black box.  We first consider the
problem of computing $\norm{f}$.

\begin{algorithm}[htbp]
  \caption{Approximate norm\label{alg:apnorm}}
  \KwIn{$T,D\in\NN$ and $\epsilon$-approximate black box for unknown $T$-sparse
    $f\in\CC[x]$ with $\deg f \leq D$}
  \KwOut{$\sigma\in\RR$, an approximation to $\norm{f}$}
  $\lambda \gets \max\left(21,\left\lceil \frac{5}{3}t(t-1)\ln d
    \right\rceil\right)$ \;
  Choose a prime $p$ randomly from $\{\lambda,\ldots,2\lambda\}$ \;
  $\omega \gets \exp(2\pi i/p)$ \;
   $w\gets (f(\omega^0),\ldots,f(\omega^{p-1}))\in\CC^p$ computed using the
   black box \;
  \KwRet{$(1/\sqrt{p})\cdot \norm{w}$}
\end{algorithm}


\begin{theorem}
  Algorithm \ref{alg:apnorm} works as stated.  On any invocation,
  with probability at least $1/2$, it returns a value
  $\sigma\in\RR_{\geq 0}$ such
  that $$(1-2\epsilon)\norm{f}<\sigma<(1+\epsilon)\norm{f}.$$
\end{theorem}
\begin{proof}
  Let $v=(f(\omega^0),\ldots,f(\omega^{p-1}))\in\CC^p$ be the vector
  of exact evaluations of $f$.  Then by the properties of our
  $\epsilon$-approximate black box we have $w=v+\epsilon\Delta$, where
  $|\Delta_i|<|f(\omega^i)|$ for $0\leq i<p$, and hence
  $\norm{\Delta}<\norm{v}$.  By the triangle inequality $\norm{w}\leq
  \norm{v}+\epsilon\norm{\Delta} < (1+\epsilon)\norm{v}$.  By
  Lemmas \ref{lem:goodp} and \ref{lem:normtrans},
  $\norm{v}=\sqrt{p}\norm{f}$ with probability at least $1/2$, so
  $(1/\sqrt{p})\cdot\norm{w}<(1+\epsilon)\norm{f}$
  with this same probability.

  To establish a lower bound on the output, note that we can make
  error in the evaluation relative to the output magnitude:
  because $\epsilon<1/2$, 
  \mbox{$|f(\omega^i)-w_i|<2\epsilon |w_i|$} for $0\leq i<p$.  We can write
  $v=w+2\epsilon\nabla$, where $\norm{\nabla}<\norm{w}$.  Then
  $\norm{v}\leq (1+2\epsilon)\norm{w}$, and
  $(1-2\epsilon)\norm{f}<(1/\sqrt{p})\cdot\norm{w}$.
\end{proof}

\subsection{Constructing an 
\texorpdfstring{{\large $\epsilon$}-}{}approximate 
remainder black box}

Assume that we have chosen a good prime $p$ for a $t$-sparse
$f\in\F[x]$.  Our goal in this subsection is a simple algorithm and
numerical analysis to accurately compute $f\rem x^p-1$. 

Assume that $f\rem x^p-1 = \sum_{0\leq i<p} b_ix^i$ exactly.  For a
primitive $p$th root of unity $\omega\in\CC$, let
$V(\omega)\in\CC^{p\times p}$ be the Vandermonde matrix built from the points
$1,\omega,\ldots,\omega^{p-1}$. Recall that
$V(\omega)\cdot(b_0,\ldots,b_{p-1})^T = (f(\omega^0),\ldots,f(\omega^{p-1}))^T$
and $V(\omega^{-1}) = p\cdot V(\omega)^{-1}$. 
Matrix vector product by such Vandermonde matrices is computed
very quickly and in a numerically stable manner by the Fast Fourier Transform (FFT).

\begin{algorithm}[htbp]
  \caption{Approximate Remainder\label{alg:apfrem}}
  \KwIn{An $\epsilon$-approximate black box for the unknown $t$-sparse
    $f\in\CC[x]$, and $p\in\NN$, a good prime for $f$}
  \KwOut{$h\in\CC[x]$ such that $\norm{(f\rem x^p-1)-h}\leq \epsilon\norm{f}$.}

  $w\gets (f(\omega^0),\ldots,f(\omega^{p-1}))\in\CC^p$ computed using the
  \hspace*{19pt} $\epsilon$-approximate black box for $f$ \;
  $u\gets (1/p)\cdot
  V(\omega^{-1})w\in\CC^p$
  using the FFT algorithm\;
\KwRet{$h=\sum_{0\leq i<p} u_ix^i\in\CC[x]$}
\end{algorithm}

\begin{theorem}
  \label{thm:apfremworks}
  Algorithm \ref{alg:apfrem} works as stated, and 
  $$\norm{\left(f\rem x^p-1\right)-h}\leq\epsilon\norm{f}.$$  It requires $O(p\log p)$ floating
  point operations and $p$ evaluations of the black box.
\end{theorem}
\begin{proof}
  Because $f$ and $f\rem x^p-1$ have exactly the same
  coefficients ($p$ is a good prime for $f$), they have exactly the
  same norm.  The FFT in Step 2 is accomplished in $O(p\log p)$
  floating point operations.  This algorithm is numerically 
  stable since $(1/\sqrt{p}) \cdot V(\omega^{-1})$ is unitary.  That is,
  assume $v=(f(\omega_0),\ldots,f(\omega^{p-1}))\in\CC^p$ is the vector
  of \emph{exact} evaluations of $f$, so
  $\norm{v-w}\leq\epsilon\norm{v}$ by the black box specification.
  Then, using the fact that $\norm{v}=\sqrt{p}\norm{f}$,
  \begin{gather*}
    \norm{(f\rem x^{p-1})-h}=
    \norm{\frac{1}{p}V(\omega^{-1})v-\frac{1}{p}V(\omega^{-1})w}\\
    = \frac{1}{\sqrt{p}} \norm{\frac{1}{\sqrt{p}}V(\omega^{-1})\cdot(v-w)}
    = \frac{1}{\sqrt{p}}\norm{v-w}\leq \frac{\epsilon}{\sqrt{p}}\norm{v} =
    \epsilon\norm{f}. \qedhere
  \end{gather*}
\end{proof}

\subsection{Creating 
\texorpdfstring{{\large $\epsilon$}-}{approximate }diversity}

First, we extend the notion of polynomial diversity to the approximate
case.

\begin{definition}
  Let $f\in\CC[x]$ be a $t$-sparse polynomial as in \eqref{eq:f} and
  $\delta\geq\epsilon>0$ such that $|c_i|\geq\delta\norm{f}$ for $1\leq
  i\leq t$.  The polynomial $f$ is said to be
  \emph{$\epsilon$-diverse} if and only if every pair of distinct
  coefficients is at least $\epsilon\norm{f}$ apart.  That is, for
  every $1\leq i < j \leq t$, $|c_i-c_j| \geq \epsilon\norm{f}$.
\end{definition}

Intuitively, if $(\epsilon/2)$ corresponds to the machine precision, this
means that an algorithm can reliably distinguish the coefficients of a
$\epsilon$-diverse polynomial. We now show how to choose a random
$\alpha$ to guarantee $\epsilon$-diversity.

\begin{theorem}
  \label{thm:epsdiv}
  Let $\delta\geq\epsilon>0$ and $f\in\CC[x]$ a $t$-sparse polynomial
  whose non-zero coefficients are of magnitude at least
  $\delta\norm{f}$.  If $s$ is a prime satisfying $s > 12$ and
  \[
  t(t-1) \leq s \leq 3.1 \frac{\delta}{\epsilon},
  \]
  then for $\zeta=\ee^{2\pi \ii / s}$ an $s$-PRU and 
  $k\in\NN$ chosen uniformly at random from $\{0,1,\ldots,s-1\}$,
  $f(\zeta^k x)$ is $\epsilon$-diverse with probability at least 
  $\frac{1}{2}$.
\end{theorem}
\begin{proof}
  For each $1\leq i\leq t$, write the coefficient $c_i$ in polar
  notation to base $\zeta$ as $c_i = r_i \zeta^{\theta_i}$, where
  each $r_i$ and $\theta_i$ are nonnegative real numbers and
  $r_i \geq \delta\norm{f}$.
  
  Suppose $f(\zeta^k x)$ is \emph{not} $\epsilon$-diverse. Then there
  exist indices $1\leq i < j \leq t$ such that
  $$\abs{r_i\zeta^{\theta_i}\zeta^{k e_i} - %
  r_j\zeta^{\theta_j}\zeta^{k e_j}} \leq \epsilon \norm{f}.$$
  Because $\min(r_i,r_j) \geq \delta\norm{f}$, the 
  value of the left hand side is
  at least 
  $\delta\norm{f}\cdot 
  \abs{\zeta^{\theta_i+ke_i} - \zeta^{\theta_j+ke_j}}$. Dividing out
  $\zeta^{\theta_j+ke_i}$, we get
  $$\abs{\zeta^{\theta_i-\theta_j}-\zeta^{k(e_j-e_i)}} \leq
  \frac{\epsilon}{\delta}.$$

  By way of contradiction, assume there exist distinct choices of $k$
  that satisfy the above inequality, say $k_1,k_2\in\{0,\ldots,s-1\}$.
  Since $\zeta^{\theta_i-\theta_j}$ and 
  $\zeta^{e_j-e_i}$ are a fixed powers of $\zeta$ not
  depending on the choice of $k$, this means
  $$\abs{\zeta^{k_1(e_j-e_i)} - \zeta^{k_2(e_j-e_i)}}\leq 
  2\frac{\epsilon}{\delta}.$$
  Because $s$ is prime, $e_i\neq e_j$, and  we assumed $k_1\neq k_2$,
  the left hand side is
  at least $\abs{\zeta-1}$. Observe that $2\pi/s$, the distance on the unit
  circle from 1 to $\zeta$, is a good approximation for this Euclidean
  distance when $s$ is large. In particular, since 
  $s> 12$, 
  $$\frac{\abs{\zeta-1}}{2\pi/s} >
  \frac{\sqrt{2}\left(\sqrt{3}-1\right)/2}{\pi/6},$$
  and therefore
  $\abs{\zeta-1} > 6\sqrt{2}(\sqrt{3}-1)/s > 6.2/s$, 
  which from the statement
  of the theorem is at least $2\epsilon/\delta$. This is a
  contradiction, and therefore the assumption was false; namely, there
  is at most one choice of $k$ such that the $i$'th and $j$'th
  coefficients collide.

  Then, since there are exactly $\binom{t}{2}$ distinct pairs of
  coefficients, and $s \geq t(t-1) = 2\binom{t}{2}$, $f(\zeta^kx)$ is
  diverse for at least half of the choices for $k$.
\end{proof}

We note that the diversification which maps $f(x)$ to
$f(\zeta^k x)$ and back is numerically stable since $\zeta$
is on the unit circle.

In practice, the previous theorem will be far too pessimistic. 
We therefore propose the method of Algorithm~\ref{alg:adaptdiverse}
to adaptively choose $s$,
$\delta$, and $\zeta^k$ simultaneously, given a good prime $p$.

\begin{algorithm}[Htbp]
\caption{Adaptive diversification\label{alg:adaptdiverse}}
\KwIn{$\epsilon$-approximate black box for $f$,
known good prime $p$, known sparsity $t$}
\KwOut{$\zeta,k$ such that $f(\zeta^k x)$ is $\epsilon$-diverse, or
\texttt{FAIL}}
$s \gets 1$,\qquad $\delta \gets \infty$,\qquad $f_p\gets 0$ \; 
\While{$s \leq t^2$ and 
       $\#\{\text{coeffs }c\text{ of }f_s\text{ s.t. }\abs{c}\geq\delta\}<t$}{
  $s \gets$ least prime $\geq 2s$ \;
  $\zeta \gets \exp(2\pi\ii/s)$ \;
  $k \gets$ random integer in $\{0,1,\ldots,s-1\}$ \label{step:adaptA}\;
  Compute $f_s = f(\zeta^k x) \rem x^p-1$ \;
  $\delta \gets$ least number s.t.\ all coefficients of $f_s$ at least $\delta$
  in absolute value are pairwise $\epsilon$-distinct \label{step:adaptB}\;
}
\lIf{$\delta > 2\epsilon$}{\KwRet{\texttt{FAIL}}}\;
\lElse{\KwRet{$\zeta^k$}}
\end{algorithm}

Suppose there exists a threshold $S\in\NN$ such that for all primes $s>S$, a
random $s$th primitive root of unity $\zeta^k$ makes $f(\zeta^k x)$
$\epsilon$-diverse with high probability. Then
Algorithm~\ref{alg:adaptdiverse} will return a root of unity whose order
is within a constant factor of $S$, with high probability. From the
previous theorem, if such an $S$ exists it must be $O(t^2)$, and hence
the number of iterations required is $O(\log t)$.

Otherwise, if no such $S$ exists, then we cannot diversify the
polynomial. Roughly speaking, this corresponds to the situation that
$f$ has too many coefficients with absolute value
close to the machine precision.
In this case, we can simply use the algorithm of
\cite{GS09} numerically, achieving the same stability but using a
greater number of evaluations and bit operations.
It is possible to
establish an adaptive hybrid between our algorithm and that of
\cite{GS09} by making $f$ as $\epsilon$-diverse \emph{as possible}
given our precision.  The non-zero coefficients of $f$ are clustered
into groups which are not $\epsilon$-diverse (i.e., are within
$\epsilon\norm{f}$ of each other).  We can use the symmetric
polynomial reconstruction of \cite{GS09} to extract the exponents within
each group.

\subsection{Approximate interpolation algorithm}

We now plug our $\epsilon$-approximate remainder black box, and 
method for making $f$ $\epsilon$-diverse, into our generic Algorithm
\ref{alg:generic} to complete our algorithm for approximate
interpolation.

\begin{theorem}
  Let $\delta>0$,
  $f\in\CC[x]$ with degree at most $D$ and sparsity at most $T$, 
  and suppose all nonzero
  coefficients $c$ of $f$ satisfy $\abs{c}>\delta\norm{f}$.
  Suppose also that $\epsilon<1.5\delta/(T(T-1))$, and we are given
  an $\epsilon$-approximate black box for $f$.
  Then, for any $\mu<1/2$ we have an algorithm to
  produce a $g\in\CC[x]$ satisfying the conditions of
  Theorem~\ref{thm:posed}. The algorithm succeeds with probability at
  least $1-\mu$ and uses $\softO(T^2\cdot \log(1/\mu)\cdot \log^2 D)$ 
  black box evaluations and floating point operations.
\end{theorem}
\begin{proof}
  Construct an approximate remainder black box for $f$ using
  Algorithm~\ref{alg:apfrem}. Then run Algorithm~\ref{alg:generic} using
  this black box as input.
  On step~\ref{step:diversify} of Algorithm~\ref{alg:generic}, 
  run Algorithm~\ref{alg:adaptdiverse},
  iterating steps~\ref{step:adaptA}--\ref{step:adaptB} 
  $\lceil \log_2 (3/\mu) \rceil$ times on each iteration through 
  the while loop to choose a diversifying $\alpha=\zeta^k$ with 
  probability at least $1 - \mu/3$.

  The cost comes from Theorems \ref{thm:genalg} and
  \ref{thm:apfremworks} along with the previous discussion and
  Theorem~\ref{thm:epsdiv}.
\end{proof}

Observe that the resulting algorithm is Monte Carlo, but could be made
Las Vegas by combining the finite fields zero testing algorithm
discussed in Section~\ref{sec:ffver} with the guarantees of
Theorem~\ref{thm:posed}.


%% file: impl.tex
\section{Implementation results}
\label{sec:impl}

We implemented our algorithms in C++ using the GNU Multiple Precision
Arithmetic Library (GMP, \url{http://gmplib.org/})
and Victor Shoup's Number Theory Library (NTL,
\url{http://www.shoup.net/ntl/})
for the exponent arithmetic. For comparison with the algorithm of
\cite{GS09}, we also used NTL's squarefree polynomial factorization
routines. We note that, in our experiments, the cost of integer polynomial
factorization (for Garg \& Schost) and Chinese remaindering were
always negligible.

\begin{table}[htbp]
\begin{center}\begin{tabular}{r|r|r|r|r|r}
  $\log_2 D$ & $T$ & Determ & G\&S MC & Alg 1 & Alg 1++ \\
  \hline
  12 & 10 & 3.77 & 0.03 & 0.03 & 0.01 \\
  16 & 10 & 46.82 & 0.11 & 0.11 & 0.08 \\
  20 & 10 & --- & 0.38 & 0.52 & 0.33 \\
  24 & 10 & --- & 0.68 & 0.85 & 0.38 \\
  28 & 10 & --- & 1.12 & 2.35 & 0.53 \\
  32 & 10 & --- & 1.58 & 2.11 & 0.66 \\
  \hline
  12 & 20 & 37.32 & 0.15 & 0.02 & 0.02 \\
  16 & 20 & --- & 0.91 & 0.52 & 0.28 \\
  20 & 20 & --- & 3.5 & 3.37 & 1.94 \\
  24 & 20 & --- & 6.59 & 5.94 & 2.99 \\
  28 & 20 & --- & 10.91 & 10.22 & 3.71 \\
  32 & 20 & --- & 14.83 & 16.22 & 4.24 \\
  \hline
  12 & 30 & --- & 0.31 & 0.01 & 0.01 \\
  16 & 30 & --- & 3.66 & 1.06 & 0.65 \\
  20 & 30 & --- & 10.95 & 6.7 & 3.56 \\
  24 & 30 & --- & 25.04 & 12.42 & 9.32 \\
  28 & 30 & --- & 38.86 & 19.36 & 13.8 \\
  32 & 30 & --- & 62.53 & 68.1 & 14.66 \\
  \hline
  12 & 40 & --- & 0.58 & 0.01 & 0.02 \\
  16 & 40 & --- & 8.98 & 3.7 & 1.54 \\
  20 & 40 & --- & 30.1 & 12.9 & 8.42 \\
  24 & 40 & --- & 67.97 & 38.34 & 16.57 \\
  28 & 40 & --- & --- & 73.69 & 36.24 \\
  32 & 40 & --- & --- & --- & 40.79
\end{tabular}\end{center}
\caption{Finite Fields Algorithm Timings\label{table:ff}}
\end{table}

In our timing results, ``Determ'' refers to the
deterministic algorithm as stated in \cite{GS09} and ``Alg 1'' is the
algorithm we have presented here over finite fields, without the
verification step. We also developed and implemented a
more adaptive, Monte Carlo version of these algorithms, as briefly
described at the end of Section~\ref{sec:generic}. The basic idea
is to sample modulo $x^p-1$ for just one prime $p\in\Theta(t^2\log d)$
that is good with high probability, then to search for much smaller good
primes. This good prime search starts at a lower bound of order $\Theta(t^2)$
based on the birthday problem, and finds consecutively larger primes
until enough primes have been found to recover the symmetric polynomial
in the exponents (for Garg \& Schost) or just the exponents (for our
method). The corresponding improved algorithms are referred to as ``G\&S
MC'' and ``Alg 1++'' below, respectively.

Table~\ref{table:ff} summarizes some timings for these four algorithms
over the finite field $\ZZ/65521\ZZ$. 
This modulus was chosen for convenience of implementation, although other
methods such as the Ben-Or and Tiwari algorithm might be more efficient in
this particlar field since discrete logarithms could be computed quickly.
However, observe that our algorithms (and those from \citeauthor{GS09})
have only poly-logarithmic dependence on the field size, and so will
eventually dominate.

The timings are given in seconds
of CPU time on a 64-bit AMD Phenom II 3.2GHz processor with 512K/2M/6M cache,
compiled using GCC 4.4.3 with the \verb|-O3| flag.
Note that the numbers listed reflect the \emph{base-2 logarithm} of the
degree bound and the sparsity bound for the randomly-generated test
cases.

The timings are mostly as expected based on our complexity estimates,
and also confirm our suspicion that primes of size $O(t^2)$ 
are sufficient to avoid exponent collisions.
It is satisfying but not particularly surprising to see that our 
``Alg 1++'' is the fastest on all inputs, as all the algorithms have a
similar basic structure. Had we compared to the Ben-Or and Tiwari or
Zippel's method, they would probably be more efficient for small sizes,
but would be easily beaten 
for large degree and arbitrary finite fields as their costs are 
super-polynomial.

\begin{table}[Htbp]
\begin{center}\begin{tabular}{r|r|r|r}
  Noise & Mean Error & Median Error & Max Error\\
  \hline
  $0\phantom{1^{-12}}$ & \expnot{4.440}{16} & \expnot{4.402}{16} & \expnot{8.003}{16} \\
  $\pm 10^{-12}$ & \expnot{1.113}{14} & \expnot{1.119}{14} & \expnot{1.179}{14} \\
  $\pm 10^{-9\phantom{1}}$ & \expnot{1.149}{11} & \expnot{1.191}{11} & \expnot{1.248}{11} \\
  $\pm 10^{-6\phantom{1}}$ & \expnot{1.145}{8} & \expnot{1.149}{8} & \expnot{1.281}{8}
\end{tabular}\end{center}
\caption{Approximate Algorithm Stability\label{table:approx}}
\end{table}

The implementation of the approximate algorithm uses machine
\texttt{double} precision (IEEE), the built-in C++
\texttt{complex<double>} type, and the popular 
Fastest Fourier Transform in the West (FFTW, \url{http://www.fftw.org/})
package for
computing FFTs. Our stability results are summarized in
Table~\ref{table:approx}. Each test case was randomly generated with
degree at most $2^{20}$ and at most 50 nonzero terms. 
We varied the precision as specified in the table and ran 10 tests in
each range. Observe that the error in our results was often \emph{less}
than the $\epsilon$ error on the evaluations themselves.

Both implementations are released under an MIT-style licence and are
available from the second author's website at\\
\url{http://www.cs.uwaterloo.ca/~droche/diverse/}.


%% file: concl.tex
\section{Conclusions}

We have shown how to use the idea of diversification to improve the
complexity of sparse interpolation over large finite fields by a factor
of $t$, the number of nonzero terms. We achieve a similar complexity for
approximate sparse interpolation, and provide the first provably
numerically stable algorithm for this purpose. Our experiments confirm
these theoretical results.

Numerous open problems remain. A primary shortcoming of our algorithms
is the quadratic dependence on $t$, as opposed to linear in the case
of dense interpolation or even sparse interpolation in smaller or
chosen finite fields using the Ben-Or and Tiwari algorithm. It seems
that reducing this quadratic dependency will not be possible without a
different approach, because of the birthday problem embedded in the
diversification step. 
In the approximate case, a provably numerically stable algorithm for
sparse interpolation with only $O(t)$ probes is still an open
question.  And, while general backward error stability is not possible
in the high degree case, it would be interesting in the case of low
degree and many variables.


%% file: ack.tex
\section*{Acknowledgements}

We thank Reinhold Burger and \'Eric Schost for pointing out an error in
an earlier draft. The comments and suggestions of the anonymous referees
were also very helpful, in particular regarding connections to previous
results and the proof of Theorem~\ref{thm:diverseff}.